\documentclass[11pt,english]{article}
\usepackage[latin9]{inputenc}
\usepackage{geometry}
\usepackage{amsthm}
\usepackage{amsmath}
\usepackage{amssymb}

\makeatletter
\theoremstyle{plain}
\newtheorem{thm}{\protect\theoremname}[section]
  \theoremstyle{definition}
  \newtheorem{defn}[thm]{\protect\definitionname}
  \theoremstyle{plain}
  \newtheorem{lem}[thm]{\protect\lemmaname}

\newif\ifFULL
\FULLtrue

\usepackage{complexity,color,bm}
\usepackage{colortbl}

\providecommand{\E}{\mathrm{E}}

\definecolor{gray-comment}{gray}{0.5}

\theoremstyle{plain}
\newtheorem{hypothesis}{Hypothesis}

\makeatletter
\newtheorem*{rep@theorem}{\rep@title}
\newcommand{\newreptheorem}[2]{%
\newenvironment{rep#1}[1]{%
 \def\rep@title{#2 \ref{##1}}%
 \begin{rep@theorem}}%
 {\end{rep@theorem}}}
\makeatother

\newreptheorem{rcor}{Corollary}
\newreptheorem{rthm}{Theorem}

\usepackage{thmtools}
\usepackage{thm-restate}

\makeatother

\usepackage{babel}
  \providecommand{\definitionname}{Definition}
  \providecommand{\lemmaname}{Lemma}
\providecommand{\theoremname}{Theorem}

\begin{document}

\title{Honest signaling in zero-sum games is hard;\\ and lying is even
harder!}

\author{Aviad Rubinstein\thanks{UC Berkeley.
I thank Shaddin Dughmi for explaining \cite{CCDEHT15-quasipoly_signaling}. I thank Jonah Brown-Cohen, Rishi Gupta, Christos Papadimitriou, Tselil Schramm, and anonymous reviewers for helpful comments on earlier drafts.
This research was supported by Microsoft Research PhD Fellowship, as well as NSF grant CCF1408635 and by Templeton Foundation grant 3966. This work was done in part at the Simons Institute for the Theory of Computing.}}
\maketitle
\begin{abstract}
We prove that, assuming the exponential time hypothesis, finding an
$\epsilon$-approximately optimal signaling scheme in a two-player
zero-sum game requires quasi-polynomial time ($n^{\tilde{\Omega}\left(\lg n\right)}$).
This is tight by \cite{CCDEHT15-quasipoly_signaling} and resolves
an open question of Dughmi \cite{Dughmi14-planted_clique}. We also
prove that finding a multiplicative approximation is \NP-hard.

We also introduce a new model where a dishonest signaler may publicly
commit to use one scheme, but post signals according to a different
scheme. For this model, we prove that even finding a $\left(1-2^{-n}\right)$-approximately
optimal scheme is \NP-hard.
\end{abstract}
 \ifFULL 
\else 
\setcounter{page}{0} 
\thispagestyle{empty}

\newpage 
\fi

\section{Introduction}

Many classical questions in economics involve extracting information
from strategic agents. Lately, there has been growing interest within
algorithmic game theory in {\em signaling}: the study of how to
{\em reveal information} to strategic agents (see e.g. \cite{MS12-more_signaling,DIR13-more_signaling,EFGLT14-more_signaling,Dughmi14-planted_clique,CCDEHT15-quasipoly_signaling}
and references therein). Signaling has been studied in many interesting
economic and game theoretic settings. Among them, {\sc Zero-Sum Signaling}
proposed by Dughmi \cite{Dughmi14-planted_clique} stands out as a
canonical problem that cleanly captures the computational nature of
signaling. In particular, focusing on zero-sum games clears away issues
of equilibrium selection and computational tractability of finding
an equilibrium. 
\begin{defn}
[{\sc Zero-Sum Signaling} \cite{Dughmi14-planted_clique}] \label{def:signaling}Alice
and Bob play a Bayesian zero-sum game where the payoff matrix $M$
is drawn from a publicly known prior. The signaler Sam privately observes
the state of nature (i.e. the payoff matrix), and then publicly broadcasts
a signal $\varphi\left(M\right)$ to both Alice and Bob. Alice and
Bob Bayesian-update their priors according to $\varphi\left(M\right)$'s
and play the Nash equilibrium of the expected game; but they receive
payoffs according to the true $M$. Sam's goal is to design an efficient
signaling scheme $\varphi$ (a function from payoff matrices to strings)
that maximizes Alice's expected payoff.
\end{defn}
Dughmi's \cite{Dughmi14-planted_clique} main result proves that assuming
the hardness of the {\sc Planted Clique} problem, there is no additive
FPTAS for {\sc Zero-Sum Signaling}. The main open question left
by \cite{Dughmi14-planted_clique} is whether there exists an additive
PTAS. Here we answer this question in the negative: we prove that
assuming the Exponential Time Hypothesis (ETH) \cite{IP01-ETH}, obtaining
an additive-$\epsilon$-approximation (for some constant $\epsilon>0$)
requires quasi-polynomial time ($n^{\tilde{\Omega}\left(\lg n\right)}$).
This result is tight thanks to a recent quasi-polynomial ($n^{\frac{\lg n}{\poly\left(\epsilon\right)}}$)
time algorithm by Cheng et al. \cite{CCDEHT15-quasipoly_signaling}.
Another important advantage of our result is that it replaces the
hardness of {\sc Planted Clique} with a more believable worst-case
hardness assumption (see e.g.~the discussion in \cite{BKW15-best_nash}). 
\begin{thm}
[Main Result] There exists a constant $\epsilon>0$, such that assuming
ETH, approximating {\sc Zero-Sum Signaling} with payoffs in $\left[-1,1\right]$
to within an additive $\epsilon$ requires time $n^{\tilde{\Omega}\left(\lg n\right)}$.
\end{thm}
Using a similar construction, we also obtain \NP-hardness for computing
a multiplicative-$\left(1-\epsilon\right)$-approximation. Unfortunately,
in our example Alice can receives both negative and positive payoffs,
which is somewhat non-standard (but not unprecedented \cite{Das13_multiplicative_hardness})
in multiplicative approximation. One main reason that multiplicative
approximation with negative payoffs is problematic is that this is
often trivially intractable for any finite factor: Start with a tiny
additive gap, where Alice's expected payoff is $c$ in the ``yes''
case, and $s=c-\epsilon$ in the ``no'' case; subtract $\left(c+s\right)/2$
from all of Alice's payoffs to obtain an infinite multiplicative hardness.
We note, however, that the combination of negative and positive payoffs
in our construction serves only to obtain structural constraints on
the resulting equilibria; the hardness of approximation is not a result
of cancellation of negative with positive payoffs: Alice's payoff
can be decomposed as a difference of non-negative payoffs $U=U^{+}-U^{-}$,
such that it is hard to approximate Alice's optimal payoff to within
$\epsilon\cdot\E\left[U^{+}+U^{-}\right]$. Nevertheless, we believe
that extending this result to non-negative payoffs could be very interesting.
\begin{thm}
There exists a constant $\epsilon>0$, such that it is \NP-hard to
approximate {\sc Zero-Sum Signaling} to within a multiplicative
$\left(1-\epsilon\right)$ factor.
\end{thm}
Finally, we note that since all our games are zero-sum, the hardness
results for {\sc Zero-Sum Signaling} also apply to the respective
notions of additive- and multiplicative-$\epsilon$-Nash equilibrium.

\subsection{The computational complexity of lying}

As a motivating example, consider the purchase of a used car (not
a zero-sum game, but a favorite setting in the study of signaling
since Akerlof's seminal ``Market for Lemons'' \cite{Akerlof70}),
and let us focus on the information supplied by a third party such
as a mechanic inspection. The mechanic (Sam) publishes a signaling
scheme: report any problem found in a one-hour inspection. Unbeknownst
to the buyer (Bob), the mechanic favors the seller (Alice), and chooses
to use a different signaling scheme: always report that the car is
in excellent condition. Notice that it is crucial that the buyer does
not know that the mechanic is lying (and more generally, we assume
that neither party knows that the signaler is lying).

Much of the work in economics is motivated by selfish agents manipulating
their private information. Here we introduce a natural extension of
Dughmi's signaling model, where the signaler manipulates his private
information. We formalize this extension in the {\sc Zero-Sum Lying}
problem, where the signaling scheme consists of two functions  $\varphi_{\textsc{alleged}}$
(``report any problem found'') and $\varphi_{\textsc{real}}$ (``car
is in excellent condition'') from payoff matrices to signals. Sam
promises Alice and Bob to use $\varphi_{\textsc{alleged}}$, which
is what Alice and Bob use to compute the posterior distribution given
the signal (i.e. the seller and buyer look at the mechanic's report
and negotiate a price as if the state of the car is correctly reflected).
But instead Sam signals according to $\varphi_{\textsc{real}}$. 

We formally define the {\sc Zero-Sum Lying} problem below; notice
that the original {\sc Zero-Sum Signaling} (Definition \ref{def:signaling})
corresponds to the special case where we restrict $\varphi_{\textsc{real}}=\varphi_{\textsc{alleged}}$.
\begin{defn}
[{\sc Zero-Sum Lying}] Alice and Bob play a Bayesian, one-shot, zero-sum
game where the payoff matrix is drawn from a publicly known prior.
A {\em dishonest signaling scheme} consists of two (possibly randomized)
functions $\varphi_{\textsc{alleged}},\varphi_{\textsc{real}}$ from
payoff matrices to signals, that induce the following protocol:
\begin{itemize}
\item Nature draws a private payoff matrix $M\sim{\cal D}_{\textsc{nature}}$.
\item Alice and Bob observe the scheme $\varphi_{\textsc{alleged}}$ and
the signal $\sigma\triangleq\varphi_{\textsc{real}}\left(M\right)$.
(But they don't know the scheme $\varphi_{\textsc{real}}$!)
\item Alice and Bob choose a Nash equilibrium $\left(\mathbf{x};\mathbf{y}\right)$
for the zero-sum game with payoff matrix $\E\left[M'\mid\varphi_{\textsc{alleged}}\left(M'\right)=\sigma\right]$%
\footnote{When $\varphi_{\textsc{alleged}},\varphi_{\textsc{real}}$ are randomized,
we have $\sigma\sim\varphi_{\textsc{real}}\left(M\right)$ and expectation
conditioned on $\E\left[M'\mid\sigma\sim\varphi_{\textsc{alleged}}\left(M'\right)\right]$.%
}.

\begin{itemize}
\item (We assume that the support of $\varphi_{\textsc{real}}$ is contained
in the support of $\varphi_{\textsc{alleged}}$.)
\end{itemize}
\item Alice and Bob receive payoffs $\mathbf{x}^{\top}M\mathbf{y}$ and
$-\mathbf{x}^{\top}M\mathbf{y}$, respectively.
\end{itemize}
Sam's goal is to compute a pair $\left(\varphi_{\textsc{alleged}},\varphi_{\textsc{real}}\right)$
that maximizes Alice's expected payoff.
\end{defn}
\vspace{0.7cm}

In the toy-setting of a biased car inspection, the Sam's optimal strategy
was very simple. In contrast, we show that for a general distribution
over zero-sum games, it is \NP-hard to find a pair $\left(\varphi_{\textsc{alleged}},\varphi_{\textsc{real}}\right)$
that is even remotely close to optimal. Notice that this is very different
from the honest case where, as we mentioned earlier, \NP-hardness
of additive approximation is unlikely given the additive quasi-PTAS
of \cite{CCDEHT15-quasipoly_signaling}. 
\begin{thm}
\label{thm:lying-1}Approximating {\sc Zero-Sum Lying} with Alice's
payoffs in $\left[0,1\right]$ to within an additive $\left(1-2^{-n}\right)$
is \NP-hard.
\end{thm}

\subsubsection*{Further discussion of dishonest signaling}

It is important to note that the dishonest signaling model has a few
weaknesses:
\begin{itemize}
\item Alice and Bob must believe the dishonest signaler. (See also further
discussion below.)
\item In particular, Sam cheats in favor of Alice, but Alice doesn't know
about it --- so what's in it for Sam? Indeed, we assume that Sam has
some intrinsic interest in Alice winning, e.g. because Sam loves Alice
or owns some of her stocks.
\item The game for which players' strategies are at equilibrium may be very
different from the actual game. Note, however, that this is also the
case for the honest signaling model (when the signaling scheme is
not one-to-one).
\item The players may receive different payoffs for different equilibria;
this may raise issues of equilibrium selection.
\end{itemize}
Despite those disadvantages, we believe that our simple model is valuable
because it already motivates surprising results such as our Theorem
\ref{thm:lying-1}. On a higher level, we hope that it will inspire
research on many other interesting aspects on dishonest signaling.
For example, notice that in our model Sam lies without any reservation;
if, per contra, the game was repeated infinitely many times, one would
expect that Alice and Bob will eventually stop believing the signals,
hence only honest signaling is possible. There is also a spectrum
of intermediate situations, where Alice and Bob observe some partial
information about past games (e.g. marginal distribution of signals)
and may encounter questions about distribution testing. 

Another related direction of potential future research is to think
about Sam's incentives. When is honest signaling optimal for Sam?
When is it approximately optimal? How should one design an effective
``punishing'' mechanism?

\subsection{Concurrent work of Bhaskar et al.}

In independent concurrent work by Bhaskar et al.~\cite{BCKS-planted_clique},
quasi-polynomial time hardness for additive approximation of {\sc Zero-Sum Signaling}
was obtained assuming the hardness of the {\sc Planted Clique} problem
(among other interesting results%
\footnote{For zero-sum games, Bhaskar et al.~also rule out an additive FPTAS
assuming $\P\neq\NP.$ This result follows immediately from our Theorem
\ref{thm:multiplicative}.%
} about network routing games and security games). Although we are
not aware of a formal reduction, hardness of {\sc Planted Clique}
is a qualitatively stronger assumption than ETH in the sense that
it requires average case instances to be hard. Hence in this respect,
our result is stronger.

\subsection{\label{sub:Techniques}Techniques}

Our main ingredient for the quasi-polynomial hardness is the technique
of ``birthday repetition'' coined by \cite{AIM14-birthday} and recently
applied in game theoretic settings in \cite{BKW15-best_nash,BPR15-PCP-PPAD, DFS16}:
We reduce from a $2$-ary constraint satisfaction problem (2-CSP)
over $n$ variables to a distribution over $N$ zero-sum $N\times N$
games, with $N=2^{\Theta\left(\sqrt{n}\right)}$. Alice and Bob's
strategies correspond to assignments to tuples of $\sqrt{n}$ variables.
By the birthday paradox, the two $\sqrt{n}$-tuples chosen by Alice
and Bob share a constraint with constant probability. If a constant
fraction of the constraints are unsatisfiable, Alice's payoff will
suffer with constant probability. Assuming ETH, approximating the
value of the CSP requires time $2^{\tilde{\Omega}\left(n\right)}=N^{\tilde{\Omega}\left(\lg N\right)}$.

\paragraph{The challenge}

The main difficulty is that once the signal is public, the zero-sum
game is tractable. Thus we would like to force the signaling scheme
to output a satisfying assignment. Furthermore, if the scheme would
output partial assignments on different states of nature (aka different
zero-sum games in the support), it is not clear how to check consistency
between different signals. Thus we would like each signal to contain
an entire satisfying assignment. The optimal scheme may be very complicated
and even require randomization, yet by an application of the Caratheodory
Theorem the number of signals is, wlog, bounded by the number of states
of nature \cite{Dughmi14-planted_clique}. If the state of nature
can be described using only $\lg N=\tilde{\Theta}\left(\sqrt{n}\right)$
bits%
\footnote{In other words, $N$, the final size of the reduction, is an upper
bound on the number of states of nature.%
}, how can we force the scheme to output an entire assignment?

To overcome this obstacle, we let the state of nature contain a partial
assignment to a random $\sqrt{n}$-tuple of variables. We then check
the consistency of Alice's assignment with nature's assignment, Bob's
assignment with nature's assignment, and Alice and Bob's assignments
with each other; let $\tau^{A,Z},\tau^{B,Z},\tau^{A,B}$ denote the
outcomes of those consistency checks, respectively. Alice's payoff
is given by:
\[
U=\delta\tau^{A,Z}-\delta^{2}\tau^{B,Z}+\delta^{3}\tau^{A,B}
\]
for some small constant $\delta\in\left(0,1\right)$. Now, both Alice
and Bob want to maximize their chances of being consistent with nature's
partial assignment, and the signaling scheme gains by maximizing $\tau^{A,B}$.

Of course, if nature outputs a random assignment, we have no reason
to expect that it can be completed to a full satisfying assignment.
Instead, the state of nature consists of $N$ assignments, and the
signaling scheme helps Alice and Bob play with the assignment that
can be completed. 

Several other obstacles arise; fortunately some can be handled using
techniques from previous works on hardness of finding Nash equilibrium
\cite{Alt94,Daskalakis-Papadimitriou-PTAS,BPR15-PCP-PPAD}.

\section{Preliminaries}

\subsubsection*{Exponential Time Hypothesis}

\begin{hypothesis} [Exponential Time Hypothesis (ETH) {\cite{IP01-ETH}}]
3SAT takes time $2^{\Omega\left(n\right)}$.

\end{hypothesis}

\subsubsection*{PCP Theorem and CSP}
\begin{defn}
[2CSP]

$2$-CSP ($2$-ary Constraint Satisfaction Problem) is a maximization
problem. The input is a graph $G=\left(V,E\right)$, alphabet $\Sigma$,
and a constraint $C_{e}\subseteq\Sigma\times\Sigma$ for every $e\in E$. 

The output is a labeling $\varphi:V\rightarrow\Sigma$ of the vertices.
Given a labeling, we say that a constraint (or edge) $\left(u,v\right)\in E$
is {\em satisfied} if $\varphi\left(u\right),\varphi\left(v\right)\in C_{\left(u,v\right)}$.
The {\em value of a labeling} is the fraction of $e\in E$ that
are satisfied by the labeling. The value of the instance is the maximum
fraction of constraints satisfied by any assignment.\end{defn}
\begin{thm}
[{PCP Theorem \cite{Dinur07-PCP}; see e.g. {\cite[Theorem~2.11]{BRKW15-DkS}} for this formulation}]\label{thm:pcp}

Given a 3SAT instance $\phi$ of size $n$, there is a polynomial
time reduction that produces a 2CSP instance $\psi$, with size $\left|\psi\right|=n\cdot\polylog n$
variables and constraints, and constant alphabet size, such that:
\begin{description}
\item [{Completeness}] If $\phi$ is satisfiable, then so is $\psi$.
\item [{Soundness}] If $\phi$ is not satisfiable, then at most a $\left(1-\eta\right)$-fraction
of the constraints in $\psi$ can be satisfied, for some $\eta=\Omega\left(1\right)$.
\item [{Balance}] Every variable in $\psi$ participates in exactly $d=O\left(1\right)$
constraints.
\end{description}
\end{thm}

\subsubsection*{Finding a good partition }
\begin{lem}
[{Essentially \cite[Lemma~6]{BPR15-PCP-PPAD}}]\label{lem:partition}

Let $G=\left(V,E\right)$ be a $d$-regular graph and $n\triangleq\left|V\right|$.
We can partition $V$ into $n/k$ disjoint subsets $\left\{ S_{1},\dots,S_{n/k}\right\} $
of size at most $2k$ such that:
\begin{equation}
\forall i,j\,\,\,\left|\left(S_{i}\times S_{j}\right)\cap E\right|\leq8d^{2}k^{2}/n\label{eq:small_intersection}
\end{equation}

\end{lem}
\ifFULL
\begin{proof}
We assign vertices to subsets iteratively, and show by induction that
we can always maintain (\ref{eq:small_intersection}) and the bound
on the subset size. Since the average set size is less than $k$,
we have by Markov's inequality that at each step less than half of
the subsets are full. The next vertex we want to assign, $v$, has
neighbors in at most $d$ subsets. By our induction hypothesis, each
$S_{i}$ is of size at most $2k$, so on average over $j\in\left[n/k\right]$,
it has less than $4dk^{2}/n$ neighbors in each $S_{j}$. Applying
Markov's inequality again, $S_{i}$ has at least $8d^{2}k^{2}/n$
neighbors in less than a $\left(1/2d\right)$-fraction of subsets
$S_{j}$. In total, we ruled out less than half of the subsets for
being full, and less than half of the subsets for having too many
neighbors with subsets that contain neighbors of $v$. Therefore there
always exists some subset $S_{i}$ to which we can add $v$ while
maintaining the induction hypothesis.
\end{proof}
\fi

\subsubsection*{How to catch a far-from-uniform distribution}

The following lemma due to \cite{Daskalakis-Papadimitriou-PTAS} implies
that 
\begin{lem}
[Lemma 3 in the full version of \cite{Daskalakis-Papadimitriou-PTAS}]\label{lem:DP-lemma}

Let $\left\{ a_{i}\right\} _{i=1}^{n}$ be real numbers satisfying
the following properties for some $\theta>0$: (1) $a_{1}\geq a_{2}\geq\dots\geq a_{n}$;
(2) $\sum a_{i}=0$; (3) $\sum_{i=1}^{n/2}a_{i}\leq\theta$. Then
$\sum_{i=1}^{n}\left|a_{i}\right|\leq4\theta$. 
\end{lem}

\section{Additive hardness}
\begin{thm}
There exists a constant $\epsilon>0$, such that assuming ETH, approximating
{\sc Zero-Sum Signaling} with payoffs in $\left[-1,1\right]$ to
within an additive $\epsilon$ requires time $n^{\tilde{\Omega}\left(\lg n\right)}$.
\end{thm}

\subsubsection*{Construction overview}

Our reduction begins with a 2CSP $\psi$ over $n$ variables from
alphabet $\Sigma$. We partition the variables into $n/k$ disjoint
subsets $\left\{ S_{1},\dots,S_{n/k}\right\} $, each of size at most
$2k$ for $k=\sqrt{n}$ such that every two subsets share at most
a constant number of constraints.

Nature chooses a random subset $S_{i}$ from the partition, a random
assignment $\vec{u}\in\Sigma^{2k}$ to the variables in $S_{i}$,
and an auxiliary vector $\hat{b}\in\left\{ 0,1\right\} ^{\Sigma\times\left[2k\right]}$.
As mentioned in Section \ref{sub:Techniques}, $\vec{u}$ may not
correspond to any satisfying assignment. Alice and Bob participate
in one of $\left|\Sigma\right|^{2k}$ subgames; for each $\vec{v}\in\Sigma^{2k}$,
there is a corresponding subgame where all the assignments are XOR-ed
with $\vec{v}$. The goal of the auxiliary vector $\hat{b}$ is to
force Alice and Bob to participate in the right subgame, i.e. the
one where the XOR of $\vec{v}$ and $\vec{u}$ can be completed to
a full satisfying assignment. In particular, the optimum signaling
scheme reveals partial information about $\hat{b}$ in a way that
guides Alice and Bob to participate in the right subgame. The scheme
also outputs the full satisfying assignment, but reveals no information
about the subset $S_{i}$ chosen by nature.

Each player has $\left(\left|\Sigma\right|^{2k}\times2\right)\times\left(n/k\times{n/k \choose n/2k}\times\left|\Sigma\right|^{2k}\right)=2^{\Theta\left(\sqrt{n}\right)}$
strategies. The first $\left|\Sigma\right|^{2k}$ strategies correspond
to a $\Sigma$-ary vector $\vec{v}$ that the scheme will choose after
observing the random input. The signaling scheme forces both players
to play (w.h.p.) the strategy corresponding to $\vec{v}$ by controlling
the information that corresponds to the next $2$ strategies. Namely,
for each $\vec{v}'\in\Sigma^{2k}$, there is a random bit $b\left(\vec{v}'\right)$
such that each player receives a payoff of $1$ if they play $\left(\vec{v}',b\left(\vec{v}'\right)\right)$
and $0$ for $\left(\vec{v}',1-b\left(\vec{v}'\right)\right)$. The
$b$'s are part of the state of nature, and the optimal signaling
scheme will reveal only the bit corresponding to the special $\vec{v}$.
Since there are $\left|\Sigma\right|^{2k}$ bits, nature cannot choose
them independently, as that would require $2^{\left|\Sigma\right|^{2k}}$
states of nature. Instead we construct a pairwise independent distribution.

The next $n/k$ strategies correspond to the choice of a subset $S_{i}$
from the specified partition of variables. The ${n/k \choose n/2k}$
strategies that follow correspond to a gadget due to Althofer \cite{Alt94}
whereby each player forces the other player to randomize (approximately)
uniformly over the choice of subset. 

The last $\left|\Sigma\right|^{2k}$ strategies correspond to an assignment
to $S_{i}$. The assignment to each $S_{i}$ is XOR-ed entry-wise
with $\vec{v}$. Then, the players are paid according to checks of
consistency between their assignments, and a random assignment to
a random $S_{i}$ picked by nature. (The scheme chooses $\vec{v}$
so that nature's random assignment is part of a globally satisfying
assignment.) Each player wants to pick an assignment that passes the
consistency check with nature's assignment. Alice also receives a
small bonus if her assignment agrees with Bob's; thus her payoff is
maximized when there exists a globally satisfying assignment.

\subsubsection*{Formal construction}

Let $\psi$ be a 2CSP-$d$ over $n$ variables from alphabet $\Sigma$,
as guaranteed by Theorem \ref{thm:pcp}. In particular, ETH implies
that distinguishing between a completely satisfiable instance and
$\left(1-\eta\right)$-satisfiable requires time $2^{\tilde{\Omega}\left(n\right)}$.
By Lemma \ref{lem:partition}, we can (deterministically and efficiently)
partition the variables into $n/k$ subsets $\left\{ S_{1},\dots,S_{n/k}\right\} $
of size at most $2k=2\sqrt{n}$, such that every two subsets share
at most $8d^{2}k^{2}/n=O\left(1\right)$ constraints.

\paragraph{States of nature}

Nature chooses a state $\left(\hat{b},i,\vec{u}\right)\in\left\{ 0,1\right\} ^{\Sigma\times\left[2k\right]}\times\left[n/k\right]\times\Sigma^{2k}$
uniformly at random. For each $\vec{v}$, $b\left(\vec{v}\right)$
is the XOR of bits from $\hat{b}$ that correspond to entries of $\vec{v}$:
\[
\forall\vec{v}\in\Sigma^{2k}\,\,\, b\left(\vec{v}\right)\triangleq\left(\bigoplus_{\left(\sigma,\ell\right)\colon\left[\vec{v}\right]_{\ell}=\sigma}\left[\hat{b}\right]_{\left(\sigma,\ell\right)}\right).
\]
Notice that the $b\left(\vec{v}\right)$'s are pairwise independent
and each marginal distribution is uniform over $\left\{ 0,1\right\} $.

\paragraph{Strategies}

Alice and Bob each choose a strategy $\left(\vec{v},c,j,T,\vec{w}\right)\in\Sigma^{2k}\times\left\{ 0,1\right\} \times\left[n/k\right]\times{\left[n/k\right] \choose n/2k}\times\Sigma^{2k}$.
We use $\vec{v}^{A},c^{A}$, etc. to denote the strategy Alice plays,
and similarly $\vec{v}^{B},c^{B}$, etc. for Bob. For $\sigma,\sigma'\in\Sigma$,
we denote $\sigma\oplus_{\Sigma}\sigma'\triangleq\sigma+\sigma'\pmod{\left|\Sigma\right|}$,
and for vectors $\vec{v},\vec{v}'\in\Sigma^{2k}$, we let $\vec{v}\oplus_{\Sigma}\vec{v}'\in\Sigma^{2k}$
denote the entry-wise $\oplus_{\Sigma}$. When $\vec{v}^{A}=\vec{v}^{B}=\vec{v}$,
we set $\tau^{A,Z}=1$ if assignments $\left(\vec{v}\oplus_{\Sigma}\vec{w}^{A}\right)$
and $\left(\vec{v}\oplus_{\Sigma}\vec{u}\right)$ to subsets $S_{j^{A}}$
and $S_{i}$, respectively, satisfy all the constraints in $\psi$
that are determined by $\left(S_{i}\cup S_{j^{A}}\right)$, and $\tau^{A,Z}=0$
otherwise. Similarly, $\tau^{B,Z}=1$ iff $\left(\vec{v}\oplus_{\Sigma}\vec{w}^{B}\right)$
and $\left(\vec{v}\oplus_{\Sigma}\vec{u}\right)$ satisfy the corresponding
constraints in $\psi$; and $\tau^{A,B}$ checks $\left(\vec{v}\oplus_{\Sigma}\vec{w}^{A}\right)$
and $\left(\vec{v}\oplus_{\Sigma}\vec{w}^{B}\right).$ When $\vec{v}^{A}\neq\vec{v}^{B}$,
we set $\tau^{A,Z}=\tau^{B,Z}=\tau^{A,B}=0$.

\paragraph{Payoffs}

Given state of nature $\left(\hat{b},i,\vec{u}\right)$ and players'
strategies $\left(\vec{v}^{A},c^{A},j^{A},T^{A},\vec{w}^{A}\right)$
and $\left(\vec{v}^{B},c^{B},j^{B},T^{B},\vec{w}^{B}\right)$, We
decompose Alice's payoff as:
\[
U^{A}\triangleq U_{b}^{A}+U_{\text{Althofer}}^{A}+U_{\psi}^{A},
\]
where
\[
U_{b}^{A}\triangleq\mathbf{1}\left\{ c^{A}=b\left(\vec{v}^{A}\right)\right\} -\mathbf{1}\left\{ c^{B}=b\left(\vec{v}^{B}\right)\right\} ,
\]
\[
U_{\text{Althofer}}^{A}\triangleq\mathbf{1}\left\{ j^{B}\in T^{A}\right\} -\mathbf{1}\left\{ j^{A}\in T^{B}\right\} ,
\]
 and
\begin{equation}
U_{\psi}^{A}\triangleq\delta\tau^{A,Z}-\delta^{2}\tau^{B,Z}+\delta^{3}\tau^{A,B},\label{eq:U^A-additive}
\end{equation}
for a sufficiently small constant $0<\delta\ll\sqrt{\eta}$.

\subsubsection*{Completeness}
\begin{lem}
\label{lem:additive-completeness}If $\psi$ is satisfiable, there
exists a signaling scheme and a mixed strategy for Alice that guarantees
expected payoff $\delta-\delta^{2}+\delta^{3}$. \end{lem}
\begin{proof}
Fix a satisfying assignment $\vec{\alpha}\in\Sigma^{n}$. Given state
of nature $\left(\hat{b},i,\vec{u}\right)$, let $\vec{v}$ be such
that $\left(\vec{v}\oplus_{\Sigma}\vec{u}\right)=\left[\vec{\alpha}\right]_{S_{i}}$.
For each $j\in\left[n/k\right]$, let $\vec{\beta}_{j}$ be such that
$\left(\vec{v}\oplus_{\Sigma}\vec{\beta}_{j}\right)=\left[\vec{\alpha}\right]_{S_{j}}$.
(Notice that $\vec{\beta}_{i}=\vec{u}$.) The scheme outputs the signal
$\left(\vec{v},b\left(\vec{v}\right),\vec{\beta}_{1},\dots\vec{\beta}_{n/k}\right)$.
Alice's mixed strategy sets $\left(\vec{v}^{A},c^{A}\right)=\left(\vec{v},b\left(\vec{v}\right)\right)$,
picks $j^{A}$ and $T^{A}$ uniformly at random, and sets $\vec{w}^{A}=\vec{\beta}_{j^{A}}$.

Because Bob has no information about $b\left(\vec{v}'\right)$ for
any $\vec{v}'\neq\vec{v}$, he has probability $1/2$ of losing whenever
he picks $\vec{v}^{B}\neq\vec{v}$, i.e. $\E\left[U_{b}^{A}\right]\geq\frac{1}{2}\Pr\left[\vec{v}^{B}\neq\vec{v}\right]$.
Furthermore, because Alice chooses $T^{A}$ and $j^{A}$ uniformly,
$\E\left[U_{\text{Althofer}}^{A}\right]=0$.

Since $\vec{\alpha}$ completely satisfies $\psi$, we have that $\tau^{A,Z}=1$
as long as $\vec{v}^{B}=\vec{v}$ (regardless of the rest of Bob's
strategy). Bob's goal is thus to maximize $\E\left[\delta^{2}\tau^{B,Z}-\delta^{3}\tau^{A,B}\right]$.
Notice that $\left(\vec{v}\oplus_{\Sigma}\vec{w}^{A}\right)$ and
$\left(\vec{v}\oplus_{\Sigma}\vec{u}\right)$ are two satisfying partial
assignments to uniformly random subsets from the partition. In particular,
they are both drawn from the same distribution, so we have that for
any mixed strategy that Bob plays, $\E\left[\tau^{B,Z}\right]=\E\left[\tau^{A,B}\right]$.
Therefore Alice's payoff is at least
\[
\left(\delta-\delta^{2}+\delta^{3}\right)\Pr\left[\vec{v}^{B}=\vec{v}\right]+\frac{1}{2}\Pr\left[\vec{v}^{B}\neq\vec{v}\right]\geq\delta-\delta^{2}+\delta^{3}.
\]
\ifFULL \else \vspace{-1cm} \fi
\end{proof}

\subsubsection*{Soundness}
\begin{lem}
\label{lem:additive-soundness}If at most a $\left(1-\eta\right)$-fraction
of the constraints are satisfiable, Alice's maxmin payoff is at most
$\delta-\delta^{2}+\left(1-\Omega_{\eta}\left(1\right)\right)\delta^{3}$,
for any signaling scheme.\end{lem}
\begin{proof}
Fix any mixed strategy by Alice; we show that Bob can guarantee a
payoff of at least $-\left(\delta-\delta^{2}+\left(1-\Omega_{\eta}\left(1\right)\right)\delta^{3}\right)$.
On any signal, Bob chooses $\left(\vec{v}^{B},c^{B}\right)$ from
the same distribution that Alice uses for $\left(\vec{v}^{A},c^{A}\right)$.
He chooses $j^{B}$ uniformly, and picks $T^{B}$ so as to minimize
$\E\left[U_{\text{Althofer}}^{A}\right]$. Finally, for each $j^{B}$,
he draws $\vec{w}^{B}$ from the same marginal distribution that Alice
uses for $\vec{w}^{A}$ conditioning on $j^{A}=j^{B}$ (and uniformly
at random if Alice never plays $j^{A}=j^{B}$). By symmetry, $\E\left[U_{b}^{A}\right]=0$
and $\E\left[U_{\text{Althofer}}^{A}\right]\leq0$.

In this paragraph, we use Althoefer's gadget to argue that, wlog,
Alice's distribution over the choice of $j^{A}$ is approximately
uniform. In Althofer's gadget, Alice can guarantee an (optimal) expected
payoff of $0$ by randomizing uniformly over her choice of $j^{A}$
and $T^{A}$. By Lemma \ref{lem:DP-lemma}, if Alice's marginal distribution
over the choice of $j^{A}$ is $8\delta^{2}$-far from uniform (in
total variation distance), then Bob can guess that $j^{A}$ is in
some subset $T^{B}\in{\left[n/k\right] \choose n/2k}$ with advantage
(over guessing at random) of at least $2\delta^{2}$. Therefore $\E\left[U_{\text{Althofer}}^{A}\right]\leq-2\delta^{2}$;
but this would imply $\E\left[U^{A}\right]\leq-2\delta^{2}+\E\left[U_{\psi}^{A}\right]\leq\delta-2\delta^{2}+\delta^{3}$.
So henceforth we assume wlog that Alice's marginal distribution over
the choice of $j^{A}$ is $O\left(\delta^{2}\right)$-close to uniform
(in total variation distance). 

Since Alice's marginal distribution over $j^{A}$ is $O\left(\delta^{2}\right)$-close
to uniform, we have that Bob's distribution over $\left(j^{B},\vec{w}^{B}\right)$
is $O\left(\delta^{2}\right)$-close to Alice's distribution over
$\left(j^{A},\vec{w}^{A}\right)$. Therefore $\E\left[\tau^{B,Z}\right]\geq\E\left[\tau^{A,Z}\right]-O\left(\delta^{2}\right)$,
and so we also get: 
\begin{gather}
\E\left[U^{A}\right]\leq\E\left[U_{\psi}^{A}\right]\leq\delta-\delta^{2}+\delta^{3}\E\left[\tau^{A,B}\right]+O\left(\delta^{4}\right).\label{eq:E=00005BU=00005D<T^=00007BA,B=00007D}
\end{gather}

\paragraph{Bounding $\E\left[\tau^{A,B}\right]$}

To complete the proof, it remains to show an upper bound on $\E\left[\tau^{A,B}\right]$.
In particular, notice that it suffices to bound the probability that
Alice's and Bob's induced assignments agree. Intuitively, if they
gave assignments to uniformly random (and independent) subsets of
variables, the probability that their assignments agree cannot be
much higher than the value of the 2CSP; below we formalize this intuition. 

By the premise, any assignment to all variables violates at least
an $\eta$-fraction of the constraints. In particular, this is true
in expectation for assignments drawn according to Alice's and Bob's
mixed strategy. This is a bit subtle: in general, it is possible that
Alice's assignment alone doesn't satisfy many constraints and neither
does Bob's, but when we check constraints between Alice's and Bob's
assignments everything satisfied (for example, think of the 3-Coloring
CSP, where Alice colors all her vertices blue, and Bob colors all
his vertices red). Fortunately, this subtlety is irrelevant for our
construction since we explicitly defined Bob's mixed strategy so that
conditioned on each set $S_{j}$ of variables, Alice and Bob have
the same distribution over assignments. 

The expected number of violations between pairs directly depends on
the value of the 2CSP. To bound the {\em probability} of observing
at least one violations, recall that every pair of subsets shares
at most a constant number of constraints, so this probability is within
a constant factor of the expected number of violations. In particular,
an $\Omega\left(\eta\right)$-fraction of the pairs of assignments
chosen by Alice and Bob violate $\psi$. 

Finally, Alice doesn't choose $j^{A}$ uniformly at random; but her
distribution is $O\left(\delta^{2}\right)$-close to uniform. Therefore,
we have $\E\left[\tau^{A,B}\right]\leq1-\Omega\left(\eta\right)+O\left(\delta^{2}\right)$.
Plugging into (\ref{eq:E=00005BU=00005D<T^=00007BA,B=00007D}) completes
the proof.
\end{proof}

\section{\label{sec:Multiplicative-hardness}Multiplicative hardness}
\begin{thm}
\label{thm:multiplicative}There exists a constant $\epsilon>0$,
such that it is \NP-hard to approximate {\sc Zero-Sum Signaling}
to within a multiplicative $\left(1-\epsilon\right)$ factor.
\end{thm}

\subsubsection*{Construction overview}

Our reduction begins with a 2CSP $\psi$ over $n$ variables from
alphabet $\Sigma$. 

Nature chooses a random index $i\in\left[n\right]$, a random assignment
$u\in\Sigma$ for variable $x_{i}$, and an auxiliary vector $\vec{b}\in\left\{ 0,1\right\} ^{\Sigma}$.
Notice that $u$ may not correspond to any satisfying assignment.
Alice and Bob participate in one of $\left|\Sigma\right|$ subgames;
for each $v\in\Sigma$, there is a corresponding subgame where all
the assignments are XOR-ed with $v$. The optimum signaling scheme
reveals partial information about $\vec{b}$ in a way that guides
Alice and Bob to participate in the subgame where the XOR of $v$
and $u$ can be completed to a full satisfying assignment. The scheme
also outputs the full satisfying assignment, but reveals no information
about the index $i$ chosen by nature.

Alice has $\left(\left|\Sigma\right|\times2\right)\times\left(n\times n\times\left|\Sigma\right|\right)=\Theta\left(n^{2}\right)$
strategies, and Bob has an additional choice among $n$ strategies
(so $\Theta\left(n^{3}\right)$ in total). The first $\left|\Sigma\right|$
strategies correspond to a value $v\in\Sigma$ that the scheme will
choose after observing the state of nature. The signaling scheme forces
both players to play (w.h.p.) the strategy corresponding to $v$ by
controlling the information that corresponds to the next $2$ strategies.
Namely, for each $v'\in\Sigma$, there is a random bit $b\left(v'\right)$
such that each player receives a small bonus if they play $\left(v',b\left(v'\right)\right)$
and not $\left(v',1-b\left(v'\right)\right)$. The $b$'s are part
of the state of nature, and the signaling scheme will reveal only
the bit corresponding to the special $v$. 

The next $n$ strategies correspond to a choice of a variable $j\in\left[n\right]$.
The $n$ strategies that follow correspond to a hide-and-seek gadget
whereby each player forces the other player to randomize (approximately)
uniformly over the choice of $j$. For Bob, the additional $n$ strategies
induce a hide-and-seek game against nature, which serves to verify
that the scheme does not reveal too much information about the state
of nature (this extra verification was unnecessary in the reduction
for additive inapproximability). 

The last $\left|\Sigma\right|$ strategies induce an assignment for
$x_{j}$. The assignment to each $x_{j}$ is XOR-ed with $v$. Then,
the players are paid according to checks of consistency between their
assignments, and a random assignment to a random $x_{i}$ picked by
nature. (The scheme chooses $v$ so that nature's random assignment
is part of a globally satisfying assignment.) Each player wants to
pick an assignment that passes the consistency check with nature's
assignment. Alice also receives a small bonus if her assignment agrees
with Bob's; thus her payoff is maximized when there exists a globally
satisfying assignment.

\subsubsection*{Formal construction}

Let $\psi$ be a 2CSP-$d$ over $n$ variables from alphabet $\Sigma$,
as guaranteed by Theorem \ref{thm:pcp}. In particular, it is \NP-hard
to distinguish between $\psi$ which is completely satisfiable, and
one where at most a $\left(1-\eta\right)$-fraction of the constraints
can be satisfied. We denote $\left(i,j\right)\in\psi$ if there is
a constraint over variables $\left(x_{i},x_{j}\right)$.

\paragraph{States of nature}

Nature chooses a state $\left(\vec{b},i,u\right)\in\left\{ 0,1\right\} ^{\Sigma}\times\left[n\right]\times\Sigma$
uniformly at random.

\paragraph{Strategies}

Alice chooses a strategy $\left(v^{A},c^{A},j^{A},t^{A},w^{A}\right)\in\Sigma\times\left\{ 0,1\right\} \times\left[n\right]\times\left[n\right]\times\Sigma$,
and Bob chooses $\left(v^{B},c^{B},j^{B},t^{B},q^{B},w^{B}\right)\in\Sigma\times\left\{ 0,1\right\} \times\left[n\right]\times\left[n\right]\times\left[n\right]\times\Sigma$.
For $\sigma,\sigma'\in\Sigma$, we denote $\sigma\oplus_{\Sigma}\sigma'\triangleq\sigma+\sigma'\pmod{\left|\Sigma\right|}$,
and for a vector $\vec{\alpha}\in\Sigma^{n}$ we let $\left(\sigma\oplus_{\Sigma}\vec{\alpha}\right)\in\Sigma^{n}$
denote the $\oplus_{\Sigma}$ of $\sigma$ with each entry of $\vec{\alpha}$.
When $v^{A}=v^{B}=v$, we set $\tau^{A,Z}=1$ if $\psi$ contains
a constraint for variables $\left(j^{A},i\right)$, and the assignments
$\left(v\oplus_{\Sigma}w^{A}\right)$ and $\left(v\oplus_{\Sigma}u\right)$
to those variables, respectively, satisfy this constraint, and $\tau^{A,Z}=0$
otherwise. Similarly, $\tau^{B,Z}=1$ iff $\left(v\oplus_{\Sigma}w^{B}\right)$
and $\left(v\oplus_{\Sigma}u\right)$ satisfy a corresponding constraint
in $\psi$; and $\tau^{A,B}$ checks $\left(v\oplus_{\Sigma}w^{A}\right)$
with $\left(v\oplus_{\Sigma}w^{B}\right).$ When $v^{A}\neq v^{B}$,
we set $\tau^{A,Z}=\tau^{B,Z}=\tau^{A,B}=0$.

\paragraph{Payoffs}

Given players' strategies $\left(v^{A},c^{A},j^{A},t^{A},w^{A}\right)$
and $\left(v^{B},c^{B},j^{B},t^{B},q^{B},w^{B}\right)$ and state
of nature $\left(\vec{b},i,u\right)$, We decompose Alice's payoff
as:
\[
U^{A}\triangleq U_{b}^{A}+U_{\text{seek}}^{A}+U_{\psi}^{A},
\]
where
\[
U_{b}^{A}\triangleq\mathbf{1}\left\{ c^{A}=\left[\vec{b}\right]_{v^{A}}\right\} /n-\mathbf{1}\left\{ c^{B}=\left[\vec{b}\right]_{v^{B}}\right\} /n,
\]
\[
U_{\text{seek}}^{A}\triangleq2\cdot\mathbf{1}\left\{ j^{B}=t^{A}\right\} -\mathbf{1}\left\{ j^{A}=t^{B}\right\} -\mathbf{1}\left\{ i=q^{B}\right\} ,
\]
 and%
\footnote{We use $\delta^{3}\tau^{A,Z}-\delta^{4}\tau^{B,Z}+\delta^{5}\tau^{A,B}$
instead of $\delta^{1}\tau^{A,Z}-\delta^{2}\tau^{B,Z}+\delta^{3}\tau^{A,B}$
as in \ref{eq:U^A-additive}, because the square of the first coefficient
appears in the proof. We have $\left(\delta^{3}\right)^{2}\ll\delta^{5}$,
but $\delta^{2}\gg\delta^{3}$.%
}
\[
U_{\psi}^{A}\triangleq\delta^{3}\tau^{A,Z}-\delta^{4}\tau^{B,Z}+\delta^{5}\tau^{A,B},
\]
for a sufficiently small constant $0<\delta\ll\sqrt{\eta}$.

\subsubsection*{Completeness}
\begin{lem}
\label{lem:multiplicative-completeness}If $\psi$ is satisfiable,
there exists a signaling scheme, such that for every signal $\mathbf{s}$
in the support, Alice can guarantee an expected payoff of $\frac{d}{n}\left(\delta^{3}-\delta^{4}+\delta^{5}\right)$. 
\end{lem}
Notice that the {\em for every signal in the support} qualification
is different than the corresponding Lemma \ref{lem:additive-completeness}
(and there is a similar difference between Lemma \ref{lem:multiplicative-soundness}
and Lemma \ref{lem:additive-soundness}). Indeed, this is stronger
than we need for proving Theorem \ref{thm:multiplicative}, but will
come handy in Section \ref{sec:Lying}.
\begin{proof}
Fix a satisfying assignment $\vec{\alpha}\in\Sigma^{n}$. Given state
of nature $\left(\hat{b},i,u\right)$, let $v$ be such that $\left(v\oplus_{\Sigma}u\right)=\left[\vec{\alpha}\right]_{i}$,
and let $\vec{\beta}$ be such that $\left(v\oplus_{\Sigma}\vec{\beta}\right)=\vec{\alpha}$.
(Notice that $\left[\vec{\beta}\right]_{i}=u$.) The scheme outputs
the signal $\mathbf{s}\triangleq\left(v,\vec{b}_{v},\vec{\beta}\right)$.
Alice's mixed strategy sets $\left(v^{A},c^{A}\right)=\left(v,\vec{b}_{v}\right)$;
picks $j^{A}$ and $t^{A}$ uniformly at random; and sets $w^{A}=\left[\vec{\beta}\right]_{j^{A}}$.
\ifFULL 

Because Bob has no information about $\left[\vec{b}\right]_{v'}$
for any $v'\neq v$, he has probability $1/2$ of losing whenever
he picks $v^{B}\neq v$, i.e. $\E\left[U_{b}^{A}\mid\mathbf{s}\right]\geq\frac{1}{2n}\Pr\left[v^{B}\neq v\mid\mathbf{s}\right]$.
Furthermore, because Alice and nature draw $t^{A},j^{A}$ and $i$
uniformly at random, $\E\left[U_{\text{seek}}^{A}\mid\mathbf{s}\right]=0$.

Since $\vec{\alpha}$ completely satisfies $\psi$, we have that $\E\left[\tau^{A,Z}\mid\mathbf{s}\right]=\Pr\left[\left(j^{A},i\right)\in\psi\mid\mathbf{s}\right]=d/n$,
as long as $v^{B}=v$ (regardless of the rest of Bob's strategy).
Bob's goal is thus to maximize $\E\left[\delta^{4}\tau^{B,Z}-\delta^{5}\tau^{A,B}\mid\mathbf{s}\right]$.
However, since Alice's assignment and nature's assignment are drawn
from the same distribution, we have that for any mixed strategy Bob
plays, $\E\left[\tau^{B,Z}\mid\mathbf{s}\right]=\E\left[\tau^{A,B}\mid\mathbf{s}\right]$.
Finally, since $i$ is a uniformly random index (even when conditioning
on \textbf{$\mathbf{s}$)}, $\E\left[\tau^{B,Z}\mid\mathbf{s}\right]\leq d/n$.
Therefore Alice's payoff is at least
\[
\frac{d}{n}\left(\delta^{3}-\delta^{4}+\delta^{5}\right)\Pr\left[v^{B}=v\mid\mathbf{s}\right]+\frac{1}{2n}\Pr\left[v^{B}\neq v\mid\mathbf{s}\right]\geq\frac{d}{n}\left(\delta^{3}-\delta^{4}+\delta^{5}\right).
\]
\else  See full version for details. \fi
\end{proof}

\subsubsection*{Soundness}
\begin{lem}
\label{lem:multiplicative-soundness}If at most a $\left(1-\eta\right)$-fraction
of the constraints are satisfiable, then for any signaling scheme
and every signal $\mathbf{s}$ in the support, Alice's maxmin payoff
is at most $\frac{d}{n}\left(\delta^{3}-\delta^{4}+\left(1-\Omega\left(1\right)\right)\delta^{5}\right)$.\end{lem}
\begin{proof}
On any signal, Bob chooses $\left(v^{B},c^{B}\right)$ from the same
distribution that Alice uses for $\left(v^{A},c^{A}\right)$. He draws
$j^{B}$ uniformly at random, and picks $t^{B}$ and $q^{B}$ so as
to minimize $\E\left[U_{\text{seek}}^{A}\mid\mathbf{s}\right]$. Finally,
for each $j^{B}$, Bob draws $w^{B}$ from the same distribution that
Alice uses for $w^{A}$ conditioning on $j^{A}=j^{B}$ (and uniformly
at random if Alice never plays $j^{A}=j^{B}$). By symmetry, $\E\left[U_{b}^{A}\mid\mathbf{s}\right]=0$
and $\E\left[U_{\text{seek}}^{A}\mid\mathbf{s}\right]\leq0$. \ifFULL

Notice that 
\begin{equation}
\E\left[U_{\psi}^{A}\mid\mathbf{s}\right]\leq\delta^{3}\cdot\Pr\left[\left(i,j^{A}\right)\in\psi\mid\mathbf{s}\right]+\delta^{5}\cdot\frac{d}{n}.\label{eq:E=00005BU-psi=00005D<Pr=00005B(i,jA) in psi=00005D}
\end{equation}
We can assume wlog that $\Pr\left[\left(i,j^{A}\right)\in\psi\mid\mathbf{s}\right]\leq3d/n$.
Otherwise, Bob can draw $\hat{j}^{A}$ from the same marginal distribution
that Alice uses for $j^{A}$, and set $q^{B}$ to a random $\psi$-neighbor
of $\hat{j}^{A}$. This would imply $\Pr\left[i=q^{B}\mid\mathbf{s}\right]\geq\Pr\left[\left(i,j^{A}\right)\in\psi\mid\mathbf{s}\right]/d$,
and therefore
\begin{align*}
\E\left[U_{\text{seek}}^{A}\mid\mathbf{s}\right] & \leq2/n-\Pr\left[\left(i,j^{A}\right)\in\psi\mid\mathbf{s}\right]/d\\
 & \leq-\Pr\left[\left(i,j^{A}\right)\in\psi\mid\mathbf{s}\right]/3d\\
 & \leq-\E\left[U_{\psi}^{A}\mid\mathbf{s}\right],
\end{align*}
in which case $\E\left[U^{A}\mid\mathbf{s}\right]\leq0$. Therefore
by (\ref{eq:E=00005BU-psi=00005D<Pr=00005B(i,jA) in psi=00005D}),
$\E\left[U_{\psi}^{A}\mid\mathbf{s}\right]\leq4d\delta^{3}/n$.

Furthermore, we claim that, conditioned on signal $\mathbf{s}$, Alice's
optimal marginal distribution over the choice of $j^{A}$ is $O\left(\delta^{3}\right)$-close
to uniform (in total variation distance). Suppose by contradiction
that the distribution is $8d\delta^{3}$-far from uniform. Then there
exists some $\ell\in\left[n\right]$ such that $\Pr\left[j^{A}=\ell\mid\mathbf{s}\right]\geq\left(1+4d\delta^{3}\right)/n$.
If Bob always plays $t^{B}=\ell$, then $\E\left[U_{\text{seek}}^{A}\mid\mathbf{s}\right]\leq-4d\delta^{3}/n$,
and hence $\E\left[U^{A}\mid\mathbf{s}\right]\leq0$.

Since Alice's distribution over $j^{A}$ is $O\left(\delta^{3}\right)$-close
to uniform, we have that Bob's distribution over $\left(j^{B},w^{B}\right)$
is $O\left(\delta^{3}\right)$-close to Alice's distribution over
$\left(j^{A},w^{A}\right)$. Therefore, $\E\left[\tau^{B,Z}\mid\mathbf{s}\right]\geq\E\left[\tau^{A,Z}\mid\mathbf{s}\right]-O\left(\delta^{3}\right)$,
so 
\begin{eqnarray*}
\E\left[\delta^{3}\tau^{A,Z}-\delta^{4}\tau^{B,Z}\mid\mathbf{s}\right] & \leq & \left(\delta^{3}-\delta^{4}+O\left(\delta^{7}\right)\right)\cdot\underbrace{\Pr\left[\left(i,j^{A}\right)\in\psi\mid\mathbf{s}\right]}_{\frac{d}{n}\left(1+O\left(\delta^{3}\right)\right)}\\
 & \leq & \frac{d}{n}\left(\delta^{3}-\delta^{4}+O\left(\delta^{6}\right)\right)
\end{eqnarray*}

We now upper bound $\E\left[\tau^{A,B}\mid\mathbf{s}\right]$. By
the premise, any assignment to all variables violates at least an
$\eta$-fraction of the constraints. In particular, this is true in
expectation for assignments drawn according to Alice's and Bob's mixed
strategy. (Recall that we defined Bob's mixed strategy so that conditioned
on each variable, Alice and Bob have the same distribution over assignments).
Therefore, since Alice's marginal distribution over $j^{A}$ is $O\left(\delta^{3}\right)$-close
to uniform, we have that 
\begin{gather*}
\E\left[\tau^{A,B}\mid\mathbf{s}\right]\leq\left(1-\Omega\left(\eta\right)+O\left(\delta^{3}\right)\right)\cdot\Pr\left[\left(j^{A},j^{B}\right)\in\psi\mid\mathbf{s}\right]\leq\frac{d}{n}\left(1-\Omega\left(\eta\right)+O\left(\delta^{3}\right)\right).
\end{gather*}
Thus in total
\begin{gather*}
\E\left[U^{A}\mid\mathbf{s}\right]\leq\E\left[U_{\psi}^{A}\mid\mathbf{s}\right]\leq\frac{d}{n}\left(\delta^{3}-\delta^{4}+\left(1-\Omega\left(\eta\right)\right)\delta^{5}+O\left(\delta^{6}\right)\right).
\end{gather*}
 \else  See full version for details. \fi
\end{proof}

\section{\label{sec:Lying}Lying is even harder}
\begin{thm}
\label{thm:lying}Approximating {\sc Zero-Sum Lying} with Alice's
payoffs in $\left[0,1\right]$ to within an additive $\left(1-2^{-n}\right)$
is \NP-hard.
\end{thm}

\subsubsection*{Construction}

Consider the construction from Section \ref{sec:Multiplicative-hardness}
for the honest signaling problem. Lemmata \ref{lem:multiplicative-completeness}
and \ref{lem:multiplicative-soundness} guarantee that there exists
a distribution ${\cal D}_{\textsc{honest}}$ of $n\times n$ zero-sum
games and constants $c_{1}>c_{2}$ such that it is \NP-hard to distinguish
between the following:
\begin{description}
\item [{Completeness}] If $\psi$ is satisfiable, there exists a signaling
scheme $\varphi_{\textsc{honest}}$, such that for any signal in $\varphi_{\textsc{honest}}$'s
support, Alice's maxmin payoff is at least $c_{1}/n$.
\item [{Soundness}] If $\psi$ is $\left(1-\eta\right)$-unsatisfiable,
for every signaling scheme $\varphi_{\textsc{honest}}^{'}$ and every
signal in the support, Alice's maxmin payoff is at most $c_{2}/n$.
\end{description}
For {\sc Zero-Sum Lying}, we construct a hard distribution of $n\times\left(n+1\right)$
zero-sum games as follows. With probability $2^{-n}$ Alice's payoffs
matrix is of the form: 
\begin{equation}
\left(\begin{array}{cc}
-A_{\textsc{honest}}^{\top} & \begin{array}{c}
-\left(c_{1}+c_{2}\right)/2n\\
\vdots\\
-\left(c_{1}+c_{2}\right)/2n
\end{array}\end{array}\right),\label{eq:lying-matrix}
\end{equation}
where Alice chooses a row (Bob chooses a column), and $A_{\textsc{honest}}$
is an $n\times n$ matrix drawn from ${\cal D}_{\textsc{honest}}$.
In other words, Bob has to choose between receiving payoff $\left(c_{1}+c_{2}\right)/2n$,
or playing a game drawn from ${\cal D}_{\textsc{honest}}$, but with
the roles reversed. 

Otherwise (with probability $1-2^{-n}$), Alice's payoff depends only
on Bob: it is $1$ if Bob chooses any of his first $n$ actions, and
$0$ otherwise; we call this the {\em degenerate game}.

Notice that we promised payoffs in $\left[0,1\right]$, whereas (\ref{eq:lying-matrix})
has payoffs in $\left[-1,0\right]$. $\left[0,1\right]$ payoffs can
be obtained, without compromising the inapproximability guarantee,
by scaling and shifting the entries in (\ref{eq:lying-matrix}) in
a straightforward manner.

\subsubsection*{Completeness}
\begin{lem}
If $\psi$ is satisfiable, there exists a dishonest signaling scheme,
such that Alice's expected payoff is at least $1-2^{-n}$.\end{lem}
\begin{proof}
We first construct $\varphi_{\textsc{alleged}}$ as follows. Whenever
nature samples a payoff matrix as in (\ref{eq:lying-matrix}), $\varphi_{\textsc{alleged}}$
outputs the signal that $\varphi_{\textsc{honest}}$ would output
for $A_{\textsc{honest}}$. Whenever Alice and Bob play the degenerate
game, $\varphi_{\textsc{alleged}}$ outputs a special symbol $\perp$. 

When Bob observes any symbol from the support of $\varphi_{\textsc{honest}}$,
he can guarantee a payoff of $c_{1}/n>\left(c_{1}+c_{2}\right)/2n$
by playing a mix of his first $n$ strategies. Therefore he only uses
his last strategy when observing the special symbol $\perp$. 

Our true signaling scheme $\varphi_{\textsc{real}}$ always outputs
an (arbitrary) signal from the support of $\varphi_{\textsc{honest}}$,
regardless of the state of nature. With probability $1-2^{-n}$, Alice
and Bob are actually playing the degenerate game, so Alice's payoff
is $1$.
\end{proof}

\subsubsection*{Soundness}
\begin{lem}
If $\psi$ is $\left(1-\eta\right)$-unsatisfiable, then for any dishonest
signaling scheme $\left(\varphi_{\textsc{alleged}}^{'},\varphi_{\textsc{real}}^{'}\right)$,
Alice's expected payoff is negative.\end{lem}
\begin{proof}
Any signal in the support of $\varphi_{\textsc{alleged}}^{'}$ corresponds
to a mixture of the degenerate game, and the distribution induced
by some signal $\mathbf{s}^{'}$ in the support of some honest signaling
scheme $\varphi_{\textsc{honest}}^{'}$ for ${\cal D}_{\textsc{honest}}$.
In the degenerate game, Bob always prefers to play his last strategy.
For any $\mathbf{s}^{'}$, Bob again prefers a payoff of $\left(c_{1}+c_{2}\right)/2n$
for playing his last strategy over a maxmin of at most $c_{2}/n$
when playing any mixture of his first $n$ strategies. Therefore,
Bob always plays his last strategy, regardless of the signal he receives,
which guarantees him a payoff of $\left(c_{1}+c_{2}\right)/2^{n+1}n>0$. 
\end{proof}
\bibliographystyle{plain}
\bibliography{signaling}

\end{document}